\documentclass[conference]{IEEEtran}
\IEEEoverridecommandlockouts

\usepackage{amsthm}
\usepackage{amssymb}
\usepackage{amsmath}
\usepackage{mathrsfs}
\usepackage{bm}
\usepackage{bbm}
\usepackage[pdftex, backref=page]{hyperref}
\usepackage{braket}
\usepackage{mathdots}
\usepackage{mathtools}
\usepackage{enumerate}
\usepackage[shortlabels]{enumitem}
\usepackage{csquotes}
\usepackage[cal=boondox]{mathalfa}
\usepackage{graphicx}
\usepackage{stackengine}
\usepackage{scalerel}
\usepackage{tensor}       
\usepackage[dvipsnames,svgnames]{xcolor}
\usepackage{array}
\usepackage{makecell}
\newcolumntype{x}[1]{>{\centering\arraybackslash}p{#1}}
\usepackage{tikz}
\usepackage{pgfplots}
\usetikzlibrary{shapes.geometric, shapes.misc, positioning, arrows, arrows.meta, decorations.pathreplacing, decorations.pathmorphing, patterns, angles, quotes, calc}
\usepackage{booktabs}
\usepackage{xfrac}
\usepackage{siunitx}
\usepackage{centernot}
\usepackage{comment}
\usepackage{chngcntr}

\newtheorem{thm}{Theorem}
\newtheorem*{thm*}{Theorem}

\newtheorem*{prop*}{Proposition}

\newtheorem*{lemma*}{Lemma}
\newtheorem{cor}[thm]{Corollary}
\newtheorem*{cor*}{Corollary}

\newtheorem*{cj*}{Conjecture}

\newtheorem*{Def*}{Definition}

\newtheorem*{question*}{Question}

\newtheorem*{problem*}{Problem}

\theoremstyle{definition}
\newtheorem{rem}[thm]{Remark}

\newcommand{\bb}{\begin{equation}\begin{aligned}\hspace{0pt}}
		\newcommand{\bbb}{\begin{equation*}\begin{aligned}}
				\newcommand{\ee}{\end{aligned}\end{equation}}
		\newcommand{\eee}{\end{aligned}\end{equation*}}

\newcommand{\eqt}[1]{\stackrel{\mathclap{\mbox{\scriptsize #1}}}{=}}
\newcommand{\leqt}[1]{\stackrel{\mathclap{\mbox{\scriptsize #1}}}{\leq}}

\newcommand{\ketbra}[1]{\ket{#1}\!\!\bra{#1}}

\newcommand{\ketbraa}[2]{\ket{#1}\!\!\bra{#2}}

\renewcommand{\epsilon}{\varepsilon}

\newcommand{\id}{\mathbbm{1}}

\DeclareMathOperator{\Tr}{Tr}

\DeclareMathAlphabet{\pazocal}{OMS}{zplm}{m}{n}

\newcommand{\WW}{\pazocal{W}}

\usepackage{lipsum}

\newcommand{\modifica}[1]{#1}               

\newcommand{\cancella}[1]{}               

\begin{document}

\title{Zero-Error List Decoding for Classical-Quantum Channels
\thanks{LL and FG acknowledge financial support from the European Union (ERC StG ETQO, Grant Agreement no.\ 101165230).}
}

\author{\IEEEauthorblockN{Marco Dalai}
\IEEEauthorblockA{\textit{Department of Information Engineering} \\
\textit{University of Brescia}\\
Brescia, Italy \\
marco.dalai@unibs.it}
\and
\IEEEauthorblockN{Filippo Girardi}
\IEEEauthorblockA{
\textit{Scuola Normale Superiore}\\
Pisa, Italy \\
filippo.girardi@sns.it}
\and
\IEEEauthorblockN{Ludovico Lami}
\IEEEauthorblockA{
\textit{Scuola Normale Superiore}\\
Pisa, Italy \\
ludovico.lami@gmail.com}
}

\maketitle

\begin{abstract}
The aim of this work is to study the zero-error capacity of pure-state classical-quantum channels in the setting of list decoding. We provide an achievability bound for list-size two and a converse bound holding for every fixed list size. The two bounds coincide for channels whose pairwise absolute state overlaps form a positive semi-definite matrix. Finally, we discuss a remarkable peculiarity of the classical-quantum case: differently from the fully classical setting, the rate at which the sphere-packing bound diverges might not be achievable by zero-error list codes, even when we take the limit of fixed but arbitrarily large list size.
\end{abstract}


\section{Introduction}

In channel coding, list decoding was introduced as a relaxation of ordinary decoding, allowing the decoder to output a list of candidate messages rather then a single one. The idea was studied extensively and has profound implications in the classical setting, both on the theoretical side and for practical applications; see, for example, Ref.~\cite{shanno-gallag-berlek-1967, forney-1968, merhav-2014,tal-vardy-2015, polyan-2016, gurusw-2007} just to mention a few. Extensions to the quantum setting have been rare, with only few results available in the literature; see, for example, Ref.~\cite{hayash-2006,bergamaschi2024,leung-smith-2008}. In this paper, we consider the zero-error setting. In the classical case, the study of zero-error list codes was initiated by Elias \cite{elias-1988} and has since then been of interest also due to connections with theoretical computer science; see \cite{fredma-komlos-1984,korner-1986,ahlswe-cai-zhang-1996,dellaf-costa-dalai-2022}. In the classical-quantum setting, several results on zero-error codes with \emph{ordinary} decoding have been obtained in recent years \cite{medeir-alleau-cohen-2006, beigi-2010, cubitt-leung-matthe-2011, duan-severi-winter-2013, duan-winter-2016}. However, the zero-error \emph{list-decoding} problem has not been investigated yet. In this paper, we derive some first results on the zero-error list-decoding capacity of pure-state channels.

\subsection{Notation and preliminaries.} 
Let $\mathcal{H}$ be a finite-dimensional Hilbert space, and let $\mathcal{D}(\mathcal{H})$ be the set of density matrices (or \emph{states}), i.e.\ positive semi-definite operators $\rho:\mathcal{H}\to \mathcal{H}$ with trace one. Let $\mathcal{X}$ be a finite set of input symbols. A \emph{classical-quantum (CQ) channel} $\WW$ is a map $\WW:\mathcal{X}\to \mathcal{D}(\mathcal{H})$; we simply write it as $x\mapsto \rho_x$. If, for all $x\in\mathcal{X}$, the state $\rho_x$ is a rank-one projector, namely $\rho_x=\ketbra{\psi_x}$ for some $\ket{\psi_x}\in \mathcal{H}$, we say that $\WW$ is a \emph{pure-state CQ channel}, which is uniquely identified by the map $x\mapsto \ket{\psi_x}$. When considering $n$ uses of a pure-state CQ channel, we assume the associated finite input set to be $\mathcal{X}^n$, the output set to be $\mathcal{D}(\mathcal{H}^{\otimes n})$ and the map to be $\WW^{\otimes n}$ defined for any $\bm{x}=(x_1,x_2,\ldots,x_n)\in\mathcal{X}^n$ by
\vspace{-3pt}
\bb
    \bm{x} \quad \mapsto \quad & \rho_{\bm{x}} \coloneqq  \rho_{x_1}\otimes \rho_{x_2}\otimes \dots\otimes \rho_{x_n}\, .
\ee

Let $M\geq 1$ and $L\geq 1$ be two integers. We use the notation $[M]=\{1,2\ldots,M\}$ and define
\bb
\mathscr{P}_L([M])\coloneqq \{S\subseteq[M]:\,|S|\leq L\}
\ee
Given a CQ channel $\WW$, an $L$-list coding scheme \modifica{with $M$ codewords and block-length $n$ is defined by a codebook $\mathcal{C}=\{\bm{x}_1,\bm{x}_2,\ldots, \bm{x}_M\}\subset \mathcal{X}^n$ and a POVM with operators $E_{\ell}$, $\ell\in\mathscr{P}_L([M])$. The coding scheme is zero-error if 
$$
\sum_{\ell\not \ni i} \Tr[\rho_{\bm{x}_i}E_\ell]=0\,,\qquad \forall i=1\ldots M\,. 
$$
For $L\geq 1$, 
the \emph{zero-error $L$-list classical capacity of $\WW$} is 
\vspace{-3pt}
\bb
C_{0,L}(\WW)\coloneqq \limsup_{n\to\infty} \max \frac{1}{n}\log \frac ML\,, 
\label{L_list_ZE_capacity}
\ee
the maximum being over $L$-list zero-error codes of length $n$.
}

\subsection{Sphere-packing bound and list-decoding.} 
For $L=1$, the quantity in~\eqref{L_list_ZE_capacity} is the ordinary zero-error capacity $C_0(\WW)$ of the classical-quantum channel. An upper bound on this quantity was discussed in~\cite{dalai-2013} as a consequence of the sphere packing bound; it takes the form
\bb
C_0(\WW) \coloneqq C_{0,1}(\WW) \leq R_\infty(\WW)\, ,
\ee
where $R_\infty$ is the rate at which the sphere-packing curve diverges. This was shown in~\cite[eq. (66)]{dalai-2013} to be expressible in the form 
\vspace{-5pt}
\bb
R_\infty (\WW) & = \min_{\sigma\in\mathcal{D}(\mathcal{H})} \max_x \; \log\frac{1}{\Tr[\Pi_{\rho_x}\sigma]}\, ,
\label{eq:Rinftyminmax}
\ee
where $\Pi_{\rho_x}$ is the projector onto the support of $\rho_x$.
Although this was not mentioned in~\cite{dalai-2013}, precisely as in the classical case~\cite{shanno-gallag-berlek-1967}, the sphere packing bound also holds for list decoding, and $R_\infty(\WW)$ is thus also a bound on the zero-error capacity with list decoding $C_{0,L}(\WW)$ for any $L$. We provide here a self-contained proof of this statement.

\begin{thm}\label{th:C_{0L}<=Rinfty} For any classical-quantum channel $\WW$ and any $L\geq 1$, we have
	\vspace{-5pt}
    \bb
        C_{0,L}(\WW)\leq R_\infty (\WW)\, .
    \ee
\end{thm}

\begin{proof}
Consider a code with $M$ codewords $\bm{x}_1, \bm{x}_2,\ldots, \bm{x}_M$ and an associated POVM \modifica{with operators $\{E_\ell\}_{\ell \in \pazocal{P}_L([M])}$}. Defining, for all $i=1,\ldots, M$, the operator
\vspace{-5pt}
\bb
F_{i}=\sum_{\ell:\, \ell \ni i } E_\ell\, ,
\ee
the zero-error assumption \modifica{coincides with
$\Tr[\rho_{\bm{x}_i}F_i]=1$ for all $i=1,\dots, M\,$},
which in turn implies
\bb
F_i\geq \Pi_{\rho_{\bm{x}_i}}\qquad \forall\, i=1,\dots, M\,.
\label{eq:F_vsPi_i}
\ee
Note that, since $|\ell|\leq L$ \modifica{for all} $\ell \in \pazocal{P}_L([M])$, we have 
\modifica{
\bb
\sum_{i=1}^M F_{i} = \sum_{i=1}^M \sum_{\ell:\, \ell \ni i } E_\ell = \sum_{\ell} \sum_{i\in\ell} E_\ell \leq \sum_{\ell} L E_\ell = L\id\,.
\ee
}
Let $\sigma$ be the minimizer in~\eqref{eq:Rinftyminmax}. From the previous equation we deduce that for at least one $i$ we must have
\bb\label{eq:Tr(E_iF)UB}
\Tr\big[F_i\sigma^{\otimes n}\big] \leq L/M\,.
\ee
On the other hand, equation~\eqref{eq:Rinftyminmax} implies that 
\bb
\Tr\big[\Pi_{\rho_{\bm{x}_i}}\sigma^{\otimes n}\big] \geq 2^{-n R_\infty(\WW)}\,.
\ee
Due to~\eqref{eq:F_vsPi_i}, we deduce that $\Tr\big[F_i \sigma^{\otimes n}\big]\geq 2^{-n R_\infty(\WW)}$. Combining this with equation~\eqref{eq:Tr(E_iF)UB}, we find that $M\leq L\cdot 2^{n R_\infty(\WW)}$. This implies the statement of the theorem.
\end{proof}

In the classical case, that is, when all the $\rho_x$ commute,
{the study of $C_{0,L}$ dates back to Elias \cite{elias-1988}. In that case, generally speaking, the situation is as follows: for any given channel $\pazocal{W}$,
\setlist[itemize]{leftmargin=4.5mm}
\begin{itemize}
    \item deducing whether $C_{0,L}(\pazocal{W})$ is positive or not is trivial;
    \item for $L=1$, the problem reduces to the so-called graph-capacity problem for which only bounds are knonw in general, except for some special cases (see \cite{lovasz-1979});
    \item for fixed $L>1$ much less is known (see e.g.\ \cite{dellaf-costa-dalai-2024});
    \item Theorem \ref{th:C_{0L}<=Rinfty} is tight, i.e.\ $\!\displaystyle{\lim_{L\to\infty}C_{0,L}(\WW) = R_\infty(\WW)}$ \cite{elias-1988}.
\end{itemize}
	
In the quantum setting, the case $L=1$ is equivalent to the classical case, while essentially nothing is known for $L\geq 2$, not even a criterion to determine whether $C_{0,L}$ is positive.
In this paper, we consider pure-state channels, proving some unexpected results. We obtain for these channels an achievability for $C_{0,2}$ and a converse on $C_{0,L}$ (Theorem~\ref{thm:main} and Theorem~\ref{thm:converse}) which together imply the following results:
\begin{itemize}
    \item $C_{0,L}>0$ for $L\geq 2$ for any non-trivial pure state channel;
    \item for a class of channels which we say have \emph{positive semi-definite absolute overlaps}, $C_{0,L}$ remains constant for \mbox{$L\geq 2$}, say $C_{0,\geq 2}$ (explicitly computable);
    \item for a subset of these channels whose states form \emph{pairwise non-obtuse angles}, it holds $C_{0,\geq 2}=R_{\infty}$;
    \item for some channels not in this restricted class, instead, $C_{0,\geq 2}<R_\infty$, so that Theorem \ref{th:C_{0L}<=Rinfty} is not tight as in the classical case.
\end{itemize}  
}

\section{Main results}
Let a pure state channel $\WW$ have states $\{\ket{\psi_x}\}_{x\in\mathcal{X}}$. {We define the \emph{absolute overlap matrix} $A_{\WW}$ of $\WW$ as the $|\mathcal{X}|\times|\mathcal{X}|$ matrix
\vspace{-5pt}
\bb
 (A_{\WW})_{xx'}\coloneqq \lvert\braket{\psi_x | \psi_{x'}}\rvert\,\qquad x,x'\in \mathcal{X}.
\ee
For any distribution $P$ on $\mathcal{X}$, we define the quadratic form with matrix $A_{\WW}$
\bb\label{eq:Q_P}
Q_P(\WW) & \coloneqq P^\intercal A_\pazocal{W}P = \sum_{x,x'} P(x)P(x') \lvert\braket{\psi_x | \psi_{x'}}\rvert.
\ee	
Note that $Q_P(\WW)$ is an average absolute overlap, as it can be written as $ \mathbb{E} \lvert\braket{\psi_X | \psi_{X'}}\rvert$, where the expectation value is taken over the two independent random variables $X,X'\sim P$. We say that a classical-quantum pure-state channel $\WW$ has \emph{positive semi-definite absolute overlaps} if
\bb\label{eq:PSD}
    A_{\WW}\geq 0.
\ee
Finally, we say that $\WW$ is \emph{pairwise non-obtuse} if there exist unit-modulus coefficients $\alpha_x$, $x\in\mathcal{X}$, such that $\braket{\alpha_x\psi_x | \alpha_{x'}\psi_{x'}}$ is real and non-negative for all $x,x'\in \mathcal{X}$. Note that any such channel has positive semi-definite absolute overlaps, since in this case the entries of $(A_{\WW})_{xx'}= |\braket{\psi_x | \psi_{x'}}|$ are the inner products of a family of vectors $\{\alpha_x\ket {\psi_{x'}}\}_{x\in \mathcal{X}}$. 

\subsection{Achievability bound}

\begin{thm}\label{thm:main}
	For a classical-quantum pure-state channel $\WW$, the list-decoding zero-error capacity with list size $L=2$ is lower bounded as
	\vspace{-5pt}
	\bb\label{th:expurgmain}
	C_{0,2}(\WW)\geq \max_P \log\frac{1}{Q_P(\WW)}\,.
	\ee
\end{thm}

\begin{proof}
	The proof is based on the standard expurgation method. We build a code with $2M$ codewords of length $n$ and then we discard $M$ which do not {fulfill} a required property. Recall that $Q_P(\WW) = \mathbb{E}|\braket{\psi_X | \psi_{X'}}|$ for independent random $X, X'$ with distribution $P$. If $Q_P=1$ for all $P$, the channel is trivial and $C_{0,L}=0$, so that the statement of the theorem is already proved. So, consider a distribution $P$ such that $Q_P<1$.
	Let $\bm{x}_1, \bm{x}_2, \dots, \bm{x}_{2M}$ be $2M$ random independent sequences of i.i.d.\ symbols with distribution $P$, and let  $\ket{\psi_{\bm{x}_1}}, \ket{\psi_{\bm{x}_2}}, \dots, \ket{\psi_{\bm{x}_{2M}}}$ be the associated $2M$ states. Then, for $i \ne j$,  we have 
	\bb
	\mathbb{E}\left[ \lvert\braket{\psi_{\bm{x}_i} | \psi_{\bm{x}_j}}\rvert \right] = Q_P^n(\WW)\,.
	\ee
	Define for each $i=1,\ldots, 2M$ the quantity
	\bb
	S_i = \sum_{j:\, j \ne i} \lvert\braket{\psi_{\bm{x}_i} | \psi_{\bm{x}_j}}\rvert\,.
	\ee
	If we choose $M = \lfloor \frac{1}{4} Q_P^{-n}(\WW) \rfloor$ we have
	\bb
	\mathbb{E}[S_i] = (2M-1) Q_P^n(\WW) < \frac{1}{2}\,.
	\ee
	By Markov's inequality,
	\vspace{-5pt}
	\bb
	\mathbb{P}\left[S_i \geq 1\right] < \frac{1}{2}\,.
	\ee
	So, the expected number of indices $i$ for which $S_i \geq 1$ is at most $\frac{1}{2} \cdot 2M = M$. Thus, for at least one extraction of the random $2M$ codewords, $S_i < 1$ for at least $M$ of them. If we only keep $M$ such codewords and throw away the other ones, the {quantity} $S_i$, redefined on the subcode, will not increase for the codewords that we keep. So, we end up with $M$ codewords such that
    \bb\label{eq:nonsingular}
         \sum_{j:\, j \ne i} \lvert\braket{\psi_{\bm{x}_i} | \psi_{\bm{x}_j}}\rvert\,<1=\braket{\psi_{\bm{x}_i} | \psi_{\bm{x}_i}} \qquad \forall i \in [M].
    \ee
    Assume without loss of generality that those are the first $M$ codewords $\ket{\psi_{\bm{x}_1}}, \ket{\psi_{\bm{x}_2}} \dots \ket{\psi_{\bm{x}_M}}$; let us call $\Psi$ the matrix of the $M$ signals $\ket{\psi_{\bm{x}_1}}, \ket{\psi_{\bm{x}_2}} \dots \ket{\psi_{\bm{x}_M}}$:
        \bb\label{eq:matrix}
        \Psi \coloneqq\begin{bmatrix}
             |& & |\\
             \ket{\psi_{\bm{x}_1}} & \cdots & \ket{\psi_{\bm{x}_M}}\\
             |&&|
        \end{bmatrix}.
    \ee
    Eq. \eqref{eq:nonsingular} ensures that the Gram matrix $G\coloneqq\Psi^\dagger\Psi$ is diagonally dominant and hence non singular by Gershgorin's theorem. Therefore, the signals $\ket{\psi_{\bm{x}_i}}_{i\in [M]}$ are linearly independent and we can introduce $\{\ket{\hat{\psi}_{\bm{x}_i}}\}_{i=1,\ldots,M}$ to be the dual basis of $\{\ket{\psi_{\bm{x}_i}}\}_{i=1,\ldots,M}$, namely
\vspace{-3pt}
\bb\label{eq:dual}
\braket{\hat{\psi}_{\bm{x}_i}|\psi_{\bm{x}_j}}=\delta_{ij}\,.
\ee
The vectors of the dual basis can be explicitly constructed by taking the columns of the following matrix:
    \bb\label{eq:matrix2}
        \hat \Psi =\begin{bmatrix}
             |& & |\\
             \ket{\hat \psi_{\bm{x}_1}} & \cdots & \ket{\hat \psi_{\bm{x}_M}}\\
             |&&|
        \end{bmatrix}\coloneqq \Psi G^{-1}.
    \ee
This clearly satisfies \eqref{eq:dual}, as $\hat{\Psi}^\dagger\Psi = G^{-\dagger}\Psi^\dagger\Psi=G^{-1}G=\id$.
Let 
\bb\label{eq:PPtilde}
\tilde P\coloneqq &\;\{(i,j)\in[M]\times [M]: i< j\}\\ 
\subseteq &\;\{(i,j)\in[M]\times [M]: i\leq  j\}\eqcolon P,
\ee
and let us define the vectors
{\bb\label{eq:POVM}
\hspace*{-0.4em}\ket{\tilde{e}_r}\! \coloneqq\!\begin{cases} 
\sqrt{g_{ij}}\ket{\hat{\psi}_{\bm{x}_i}} + \sqrt{g_{ji}}\ket{\hat{\psi}_{\bm{x}_j}} &\!r=(i,j)\in \tilde P,\\
\sqrt{1-\sum_{k\neq i} |g_{ik}|} \ket{\hat{\psi}_{\bm{x}_i}}  &\! r=(i,i)\in P\setminus\tilde P.
\end{cases}
\ee}
in terms of the entries $g_{ij}$ of the Gram matrix $G$.
Then the POVM formed by the rank-one operators 
\bb
E_r\coloneqq\ketbra{\tilde{e}_{r}} \qquad r\in P,
\ee 
which outputs the list $\ell=\{r_1,r_2\}$ when the outcome is \mbox{$r=(r_1,r_2)$}, clearly gives list decoding with zero probability of error. Indeed, if we measure the states in the basis $\{\ket{e_r}\}_r$, we see that each output has positive probability for at most two input codewords. Let us verify that this is actually a POVM: the operator
{
\bb
\sum_{r\in P} \ketbra{\tilde{e}_{r}}&=\sum_{i,j:i\neq j} \big(|g_{ij}|\ketbra{\hat{\psi}_{\bm{x}_i}}+ g_{ij}\ketbraa{\hat{\psi}_{\bm{x}_i}}{\hat{\psi}_{\bm{x}_j}}\big) \\[-1mm]
&\qquad\qquad +\sum_{i}\Big(1-\sum_{k\neq i} |g_{ik}|\Big)\ketbra{\hat{\psi}_{\bm{x}_i}}\\
&= \hat{\Psi}G\hat{\Psi}^\dagger= \Psi G^{-1} \Psi^\dagger
\ee}
is the projection on the space spanned by the $\psi_{\bm{x}_i}$. 
	In the above, we have chosen $M=\lfloor \frac{1}{4} Q_P^{-n}(\WW) \rfloor$, so that our code has rate
	\vspace{-5pt}
	\begin{align}
		R & \geq - \log Q_P(\WW) -o(1)\,.
	\end{align}
	Maximizing over $P$, we obtain the statement.
\end{proof}

The algebraic reason behind the existence of the set of vectors defined in \eqref{eq:POVM} is the following. For any arbitrary  $M\times M$ diagonally dominant Hermitian matrix $G=(g_{ij})$, i.e.\
	\vspace{-5pt}
	\bb\label{eq:condition_g}
	  \qquad g_{ii}\geq \sum_{j\ne i} |g_{ij}|\,\qquad  \forall i=1,\dots, M,
	\ee
	   there exists a $\tfrac{M(M+1)}{2} \times M$ matrix $V$ with at most two non-zero elements in each row such that	
       \bb \label{eq:decomposition}
       G = V^\dagger\, V.
       \ee
Therefore, up to some isometric mapping to a larger reference system, there exists an orthonormal basis $\{\ket{e_r}\}_{r\in S}$} with respect to which the representations of the $M$ signals $\ket{\psi_{\bm{x}_1}}, \ket{\psi_{\bm{x}_2}} \dots \ket{\psi_{\bm{x}_M}}$ are the $M$ columns of $V$.
    Because each row of $V$ has at most two non-zero entry, if we measure the states in the basis $\{\ket{e_k}\}_k$, we see that each output has positive probability for at most two input codewords. So, the $M$ codewords are list-of-two decodable with zero error. 
The decomposition \eqref{eq:decomposition} is strictly connected to the notion of \emph{factor width}, and the proof relies on an easy extension of the results in \cite{BOMAN2005239}. For convenience, in the Appendix we provide a simple construction of $V$ that is consistent with the form of \eqref{eq:POVM}.
\begin{rem}
    It is interesting to compare the POVM that achieves the result of Theorem~\ref{thm:main} with the POVM used in~\cite{holevo-1998} and~\cite{holevo-2000}. Those achievability bounds essentially use projectors build from the columns of $\hat{\Psi}=\Psi G^{-1/2}$ while here we use the columns of $\hat{\Psi}=\Psi G^{-1}$ in pairs and scaled by coefficients associated to inner products. Note that $\hat{\Psi}=\Psi G^{-1}$ was also considered for the unambiguous measurement strategy studied in~\cite{takeok-krovi-guha-2013}.
\end{rem}

{
For a class of channels already considered in~\cite{Dalai_2013}, Theorem~\ref{thm:main} is tight and extends to every $L\geq 2$. Indeed, it was proved in~\cite[Th. 9]{Dalai_2013} that if a pure-state channel has the property that $\braket{\psi_x | \psi_{x'}} \ge 0$ for all $x,x'\in\mathcal{X}$, then we can rewrite $R_\infty(\WW)$ as
\bb\label{eq:Dalai_2013}
R_\infty(\WW) & = \max_P \Big[ -\log \sum_{x,x'} P(x)P(x') \braket{\psi_x | \psi_{x'}} \Big]\\
& = \max_P \log Q_P(\WW)^{-1}
\ee
Hence, in this case, the converse bound of Theorem~\ref{th:C_{0L}<=Rinfty} coincides with the achievability bound of Theorem~\ref{thm:main}. Because global phases on the pure states will not affect the probability of error in any measurement, the condition $\braket{\psi_x | \psi_{x'}} \ge 0$ can only be meaningful up to global phases of the vectors (see \cite[Rem.~3]{Dalai_2013}). As a consequence, for any arbitrary pure-state pairwise non-obtuse channel, we can replace $ \braket{\psi_x | \psi_{x'}}$ with $|\braket{\psi_x | \psi_{x'}}|$ in \eqref{eq:Dalai_2013}, obtaining the following result.

\begin{cor} 
\label{cor_posinnprod}
	For a pure-state pairwise non-obtuse channel $\WW$,
    for any $L\geq 2$ we have
    \vspace{-5pt}
    \bb
        C_{0,L}(\WW)= \max_P \log \frac{1}{Q_P(\WW)}=R_\infty(\WW).
    \ee
\end{cor}

\begin{rem} 
We point out that any binary pure-state channel is pairwise non-obtuse,
as we can always find a unit-modulus coefficient $\alpha_1$ such that $\braket{\psi_0|\alpha_1\psi_1}\geq 0$. Hence, for $L\geq2$, $C_{0,L}(\WW)=R_\infty(\WW)$ for all binary pure-state channels.
\end{rem}
}

\subsection{Converse bound}

Given a pure state channel $\WW$ with states $\{\ket{\psi_x}\}$, consider the set of matrices
\bb\label{eq:def_A}
\hspace*{-1mm}\mathcal{A}_{\WW}\coloneqq\left\{A\in \mathbb{R}^{|\mathcal{X}|\times|\mathcal{X}|} : A\geq 0\, |A_{x,x'}|\leq \lvert \braket{\psi_x|\psi_{x'}} \rvert \right\} .
\ee
Define, for any $A \in \mathcal{A}_{\WW}$, the quantity
\bb
Q_P(A) \coloneqq \sum_{x,x'} P(x)P(x') A_{x,x'}\,.	
\ee

\begin{thm}\label{thm:converse}
For any classical-quantum pure-state channel $\WW$, the list-decoding zero-error capacity with list size $L$ is upper bounded as
\vspace{-5pt}
\bb
C_{0,L}(\WW)\leq \min_{A\in\mathcal{A}_{\WW}}\max_P \log\frac{1}{Q_P(A)}\,.
\ee
\end{thm}
\begin{proof}
		Let $\ket{\psi_{\bm{x}_1}}, \ket{\psi_{\bm{x}_2}}, \dots, \ket{\psi_{\bm{x}_{M}}}$ be a code which is list-of-$L$ decodable with zero error. Then, because any POVM can be seen as a projective measurement in a larger space, we can assume that $\ket{\psi_{\bm{x}_1}}, \ket{\psi_{\bm{x}_2}}, \dots, \ket{\psi_{\bm{x}_{M}}}$ live in a large space with an orthonormal basis $\{\ket{e_k}\}_k$   such that all the operators of the measurement are projectors diagonal in that basis. We can then write our vectors in that basis and write their coordinates as columns of a matrix $V$.
    So, $G=V^\dagger V$ is the Gram matrix of the $\ket{\psi_{\bm{x}_i}}$, and each row of $V$ only contains at most $L$ non-zero entries, due to the assumption of list-of-$L$ decodability. Let {$\bra{v_k}\coloneqq \sum_{i=1}^M v_{ki} \bra{i}$}, with \mbox{$v_{ki}\coloneqq \braket{e_k|\psi_{\bm{x}_i}}$}, be the $k$-th row of $V$, so that we can write
	\bb\label{eq:31}
	V=\sum_k \ketbraa{e_k}{v_k}\,,\qquad G=\sum_k \ketbra{v_k}\,.
	\ee
	Calling $g_{ij}$ the matrix elements of $G$, we define
	\bb
	T\coloneqq \sum_{i,j}|g_{ij}|\,.
	\ee
	{Consider the contribution of each $\ket{v_k}$ to the diagonal of $T$:} 
	\bb
	D_k & \coloneqq \sum_i  {|v_{ki}|^2}\,.
	\ee
	Then we must have
	$\sum_kD_k =\sum_{i}g_{ii}=M$,
    which can be also seen if we directly compute
	\bb\label{eq:34}
	\sum_kD_k =\sum_{i}\sum_k \braket{\psi_{\bm{x}_i}|e_k}\braket{e_k|\psi_{\bm{x}_i}}=M\,.
	\ee
	Let us define the total ``absolute contribution'' given by $\ket{v_k}$ to $T$ as
	\vspace{-5pt}
	\bb
	T_k & \coloneqq \sum_{i, j} |v_{ki}v_{kj}|\,.
	\ee
	Note that, because $v_{ki}\neq 0$ for at most $L$ values of $i$, by the Cauchy--Schwarz inequality we have
	\bb\label{eq:37}
	T_k &  = \Big(\sum_i |v_{ki}|\Big)^2 = \Big(\sum_{i:\, v_{ki}\neq 0} |v_{ki}|\cdot 1\Big)^2\\
	& \leq L\sum_i |v_{ki}|^2 = L D_k\,,
	\ee
	Then, by~\eqref{eq:31}, we find
	\bb\label{eq:38}
	T & =\sum_{i, j} |g_{ij}|=\sum_{i, j} \left|\sum_k v_{ki}^*v_{kj}\right|\\
	& \leq \sum_{i, j} \sum_k |v_{ki}||v_{kj}|=  \sum_k \sum_{i, j}|v_{ki}||v_{kj}|\\
	& = \sum_k T_k\leqt{(i)} \sum_k LD_k \eqt{(ii)} LM\,,
	\ee
    where in (i) we have leveraged~\eqref{eq:37}, and in (ii) we have recalled~\eqref{eq:34}.
	Now, for any $i, j$ and any $A\in\mathcal{A}$ (see~\eqref{eq:def_A}), calling $\bm{x}_i=(x_{i1},\dots, x_{in})$ and setting $\bm{A}\coloneqq A^{\otimes n}$, we have
	\begin{align*}
	|g_{ij}| & = |\!\braket{\psi_{\bm{x}_i}|\psi_{\bm{x}_j}}\!|\geq \prod_{\ell =1}^n |A_{x_{i\ell},x_{j\ell}}|\geq \prod_{\ell =1}^n A_{x_{i\ell},x_{j\ell}}=\bm{A}_{\bm{x}_i,\bm{x}_j} .
	\end{align*}
	Then, we can lower bound
	\bb\label{eq:40}
	T  =\sum_{i, j} |g_{ij}|
	& \geq M^2 \frac{1}{M^2} \sum_{i, j} \bm{A}_{\bm{x}_i,\bm{x}_j}\\
	& \geq M^2 \min_{\bm{P}} \sum_{\bm{x}, \bm{x'}}\bm{P}(\bm{x})\bm{P}(\bm{x}') \bm{A}_{\bm{x},\bm{x}'}\,,
	\ee
    where $\bm{P}$ is an arbitrary probability distribution over $\mathcal{X}^{n}$. 
	Since $A\geq 0$, also $\bm{A}\geq 0$ and, as proved in~\cite{jeline-1968}, we have 
	\begin{align}\nonumber
	\min_{\bm{P}} \sum_{\bm{x}, \bm{x'}}\bm{P}(\bm{x})\bm{P}(\bm{x}') \bm{A}_{\bm{x},\bm{x}'} & = \Big(\min_{P} \sum_{x, x'}P(x)P(x') A_{x,x'}\Big)^n\\
	& =\left[\min_P Q_P(A)\right]^n\,.\label{eq:41}
	\end{align}
	Thus, combining~\eqref{eq:40} with~\eqref{eq:41}, we get
	\bb
	T \geq M^2 \left[\min_P Q_P(A)\right]^n\,.
	\ee
	Recalling that, by~\eqref{eq:38}, we have $T\leq LM$, we infer that
	\bb
	M\leq L \left[\min_P Q_P(A)\right]^{-n}
	\label{M_L_converse_bound}
    \ee
	Taking logarithms and optimizing over $A\in\mathcal{A}_{\WW}$ we deduce the statement of the theorem.
\end{proof}

{
\begin{rem}
	\label{remark:M/L}
	We observe that \eqref{M_L_converse_bound} implies a stronger statement than the theorem, since it holds for any finite $n$, and hence even if $L$ grows with $n$. In particular, with the rate of the code taken to be $\frac{1}{n}\log(M/L)$ as in  \cite{shanno-gallag-berlek-1967} (and coherently with \eqref{L_list_ZE_capacity}), 
	the bound of the theorem applies to the rate of zero-error codes even with exponential list-size. The same remark holds for Theorem \ref{th:C_{0L}<=Rinfty}.
\end{rem}
}

\subsection{Matching bounds}

While Corollary~\ref{cor_posinnprod} suffices to deduce a closed-form expression for $C_{0,L}$ for pairwise non-obtuse channels, Theorem~\ref{thm:converse} allows us to deduce a closed-form expression also for pure-state channels with positive semi-definite absolute overlaps. Indeed, for these channels the absolute overlap matrix $A_{\WW}$ belongs to the family $\mathcal{A}_{\WW}$. This means that, combining the achievability bound of Theorem~\ref{thm:main} with the converse bound of Theorem~\ref{thm:converse}, we immediately get
\vspace{-5pt}
\begin{align*}
\max_P \log\frac{1}{Q_P(\WW)}&\leq C_{0,2}(\WW)\\
&\leq C_{0,L}(\WW)\leq \min_{A\in\mathcal{A}_{\WW}}\max_P \log\frac{1}{Q_P(A)}\\
&\phantom{\leq C_{0,L}(\WW) \;}\leq \max_P \log\frac{1}{Q_P(A_{\WW})}\,.
\end{align*}
Then, by the simple remark that $Q_P(\WW)=Q_P(A_{\WW})$, we get the following corollary.

\begin{cor}\label{cor:matching}
	Let a pure-state classical-quantum channel $\WW$ has positive semi-definite absolute overlaps, as in~\eqref{eq:PSD}. Then, for any $L\geq 2$,
	\vspace{-5pt}
	\bb
	C_{0,L}(\WW)= \max_P \log\frac{1}{Q_P(\WW)}\,.
    \label{eq:C0LPSDoverlaps}
	\ee
	where $Q_P$ is defined in~\eqref{eq:Q_P}. 
\end{cor}

\begin{rem} Corollary~\ref{cor:matching} holds whenever $|\mathcal{X}|\leq 3$; indeed, it can be proved using Sylvester's criterion that any $3\times 3$ positive semi-definite matrix remains positive semi-definite after taking the entrywise absolute value.
\end{rem}

While for pairwise non-obtuse channels Corollary~\ref{cor_posinnprod} also shows that $C_{0,L}=R_\infty$, Theorem~\ref{thm:converse} does not allow us to deduce the same for more general channels with only positive semi-definite absolute overlaps. The next section shows that this is not a limitation of the proof.

\section{The Trine Channel}\label{sec:trine_conv_II}

We show here that, in general, for classical-quantum channels the rate $R_\infty$ might not be achievable by zero-error list codes even with $L\to \infty$, differently from the fully classical setting. The example also also shows that Corollary~\ref{cor_posinnprod} does not hold in general for channels with positive semi-definite absolute overlaps. To show this, we use a channel $\pazocal T$ with states taken from the ubiquitous trine set of Figure~\ref{fig:trine}; see~\cite{holevo-1973} and~\cite{peres-wootte-1991}.
\begin{figure}
	\centering
	\includegraphics[scale=0.9]{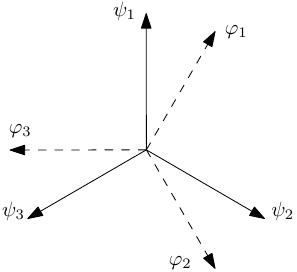}
	\caption{Trine channel. All the vectors lie in a plane; $\psi_1$, $\psi_2$ and $\psi_3$ form angles of $120^\circ$, and each vector $\varphi_i$ is orthogonal to $\psi_i$.}
	\label{fig:trine}
\end{figure}
We first consider the zero-error capacity with list decoding. Since
\vspace{-5pt}
\bb
    A_{\pazocal T}=
    \begin{pmatrix}
        1 & 1/2 & 1/2 \\
        1/2 & 1 & 1/2 \\
        1/2 & 1/2 & 1
    \end{pmatrix}
\ee
is positive semi-definite, Corollary~\ref{cor:matching} applies. By symmetry of the channel under any permutation of the states, and by convexity  in $P$ of the quadratic form in~\eqref{eq:Q_P} (because $A_{\pazocal T}$ is positive semi-definite), the maximizing $P$ in~\eqref{eq:C0LPSDoverlaps} is the uniform distribution, and $C_{0,L}(\pazocal T)=\log \frac{3}{2}$ for any $L\geq 2$. 

We then observe that for this channel the sphere packing bound reduces to a sharp threshold. Setting $\rho_P=\mathbb{E}[\ketbraa{\psi_x}{\psi_x}]$, for uniform $P$ we find $\rho_P=\frac{1}{2}\id$, so that $H(\rho_P)=1$ and $\lambda_{\max}(\rho_P)=1/2$. This implies that $C(\pazocal T)=R_\infty(\pazocal T)=1$. So, for this channel, 
\vspace{-5pt}
\bb
C_{0,L}(\pazocal T)=\log (3/2)<1=R_\infty(\pazocal T)\quad \forall\, L\geq 2\,.
\ee
{Actually, as observed in Remark \ref{remark:M/L}, for this channel the rate $\frac{1}{n}\log(M/L)$ of zero-error codes is upper bounded by $\log (3/2)$ even with exponential list-size.}

As a side remark, we point out that the rate $\log(3/2)$ can also be achieved by zero-error codes with a separable measurement if we allow $L\to\infty$. Indeed, considering vectors $\{\ket{\varphi_x}\}$ as shown in Figure~\ref{fig:trine}, if we apply a measurement with operators $\{\frac{2}{3}\ket{\varphi_x}\bra{\varphi_x}\}$ we can reduce the channel to a classical ternary typewriter channel, for which $R_\infty=\log(3/2)$. Then, this rate is achievable by zero-error codes with list decoding as the list size $L$ grows unbounded \cite{elias-1988}. We point out that, although the proof in~\cite{elias-1988} uses $L\to\infty$, no converse is known showing that $L=2$ does not suffice to achieve $\log(3/2)$ for the ternary typewriter channel. See~\cite{bhanda-khetan-2024a} for the best known result. 

\appendix\label{ap:appendix}
    \begin{proof}[Proof of the decomposition~\eqref{eq:decomposition}]
Let $\tilde P$ and $P$ as in \eqref{eq:PPtilde}. Then we define the matrix $\tilde V=(\tilde V_{rk})_{r\in \tilde R, k\in[M]}$ as follows:
\bb
    \tilde V_{rk}\coloneqq \sqrt{g_{ji}}\delta_{ki}+\sqrt{g_{ij}}\delta_{kj} \quad \text{where}\quad r=(i,j).
\ee
By construction, $\tilde V$ has at most two non-zero elements in each row. Furthermore, using $g_{ij}=g_{ji}^\ast$, it is easy to show that
\begin{align} 
    (\tilde V^\dagger \tilde V )_{k'k}
    &=\begin{cases}
        g_{k'k} & k'\neq k\\
        \sum_{i\neq k}|g_{ik}| & k'=k
    \end{cases}
\end{align}
By \eqref{eq:condition_g},  $\sqrt{g_{kk}-(\tilde V^\dagger \tilde V )_{kk}}$ is real and non-negative for all $k\in[M]$. Then, the matrix $V=(V_{rk})_{r\in P, k\in [M]}$ defined as
\bb\label{eq:V}
    V_{rk}\coloneqq \begin{cases}
        \tilde V_{rk} & r\in \tilde P\\
        \sqrt{g_{kk}-\sum_{i\neq k}|g_{ik}|}\;\delta_{r,(k,k)} & r\in P\setminus \tilde P
    \end{cases}
\ee
clearly satisfies the property $V^\dagger V = G$. Moreover, $V$ is an extension of $\tilde V$ obtained by adding rows with at most one non-zero element. Since $|P|=M(M+1)/2$, we have completed the proof.
\end{proof}

\bibliographystyle{IEEEtran}
\bibliography{biblio_MD,biblio }

\end{document}